\documentclass[12pt]{article}
\usepackage{amssymb}
\usepackage{amsmath}
\usepackage{amsthm}
\usepackage{calc}
\usepackage{xcolor}
\usepackage{cancel}
\usepackage[affil-it]{authblk}
 
\usepackage{graphicx}

\newtheorem{theorem}{Theorem}
\newtheorem{prop}[theorem]{Proposition}

\newtheorem{remark}[theorem]{Remark}
\newenvironment{rem}{\begin{remark} \rm}{\end{remark}}
\newtheorem{example}[theorem]{Example}

\newcommand{\den}[2]{{{\delta^{(n)}_{#1,#2}}}}

\def\d{{\mathrm d}}

\usepackage{cancel}

\def\parpo#1#2{\{#1,#2\}}

\newcommand{\RR}{{{\mathbb{R}}}}
\newcommand{\brke}[2]{{{\langle #1,#2\rangle}}}

\def\M{{\mathcal M}}
\def\Tr{\operatorname{Tr}}

\begin{document}
\title{Poisson quasi-Nijenhuis manifolds, closed Toda lattices,\\
and generalized recursion relations}
\date{} 

\author{E.\ Chu\~no Vizarreta${}^{1}$, G.\ Falqui${}^{2,5}$, I.\ Mencattini${}^3$, 
M.\ Pedroni${}^{4,5}$}

\affil{
{\small  $^1$Unidade Acad\^emica de Belo Jardim, Universidade Federal Rural de Pernambuco, Brazil}\\
{\small eber.vizarreta@ufrpe.br  
}\\
\medskip
{\small  $^2$Dipartimento di Matematica e Applicazioni, Universit\`a di Milano-Bicocca, Italy
}\\
{\small  gregorio.falqui@unimib.it 
}\\
\medskip
{\small $^3$Instituto de Ci\^encias Matem\'aticas e de Computa\c c\~ao,  Universidade de S\~ao Paulo, Brazil}\\
{\small igorre@icmc.usp.br 
}\\
\medskip
{\small $^4$Dipartimento di Ingegneria Gestionale, dell'Informazione e della Produzione,  Universit\`a di Bergamo, Italy}\\
{\small marco.pedroni@unibg.it 
}\\
\medskip
{\small  $^5$INFN, Sezione di Milano-Bicocca, Piazza della Scienza 3, 20126 Milano, Italy}
}

\maketitle
\abstract{\noindent
We present two involutivity theorems in the context of Poisson quasi-Nijenhuis 
manifolds. 
The second one stems from recursion relations that generalize the so called Lenard-Magri relations on a bi-Hamiltonian manifold. We apply these results to the closed (or periodic) Toda lattices of type 
$A_n^{(1)}$, $C_n^{(1)}$, $A_{2n}^{(2)}$ and, for the ones of type $A^{(1)}_n$, we show how this geometrical setting relates to their bi-Hamiltonian representation and to their recursion relations.}

\medskip\par\noindent
{\bf Keywords:} Integrable systems; Toda lattices; Poisson quasi-Nijenhuis manifolds; bi-Hamiltonian manifolds; Flaschka coordinates.
\medskip\par\noindent
{\bf Mathematics Subject Classification:} 37J35, 53D17, 70H06.  

\baselineskip=0,6cm

\section{Introduction}
Poisson-Nijenhuis (PN) manifolds \cite{MagriMorosiRagnisco85,KM} were introduced to geometrically describe 
the properties of Hamiltonian integrable systems. Such manifolds are endowed with a Poisson tensor $\pi$ 
and with a tensor field $N$ of type $(1,1)$, sometimes called ``recursion" or ``Nijenhuis" operator (see \cite{BKM2022} and references therein),
which is torsionless and compatible (see Section \ref{sec:PN-PqN}) with 
$\pi$. They are important examples of bi-Hamiltonian manifolds, with the traces $H_k$ of the powers 
of $N$ satisfying the Lenard-Magri relations and thus being in involution with respect to both Poisson brackets induced by the Poisson tensors. Such $H_k$ 
can be considered as prototypical examples of geometrically defined Liouville integrable systems.

An interesting generalization of the notion of PN manifold was introduced in \cite{SX}. In that paper, a Poisson quasi-Nijenhuis (PqN) manifold was defined to be a Poisson manifold with a compatible tensor field $
N$ of type $(1,1)$, whose torsion needs not vanish but is controlled by a suitable 3-form $\phi$, as we recall in Section \ref{sec:PN-PqN}. 
This means that, in general, the Lenard-Magri relations do not hold and so the traces ${
H}_k$ of the powers of ${
N}$ are not in involution.
For this reason, no application of PqN manifolds to the theory of 
integrable systems was found until the paper \cite{FMOP2020}, where 
sufficient conditions
entailing that the functions ${
H}_k$ are in involution were found.
In the same paper, these results were applied to interpret the well known integrability of the closed Toda lattice in the PqN framework, showing that its integrals of motion are the traces of the powers of a tensor field 
of type $(1,1)$, which is a deformation (by means of a suitable closed 2-form $\Omega$) 
of the Das-Okubo recursion operator \cite{DO} 
of the open Toda lattice. 
Further general results about this deformation process were presented in \cite{FMP2023,DMP2024}.
In particular, a relevant notion therein introduced and discussed is the possibility of  ``bouncing" between PN and PqN manifolds by deforming a Nijenhuis operator by means of a judiciously chosen closed $2$-form $\Omega$.

In this paper we prove a stronger version of the above mentioned involutivity theorem, Theorem \ref{thm:involution}, concerning PqN manifolds which are not necessarily deformations of PN manifold. Moreover, we present a second involutivity theorem, Theorem \ref{thm:involution2}, exploiting some generalized 
recursion relations involving both the $-\Omega$ deformation and the traces of the powers of the $+\Omega$ deformation of a given PN structure. 
These results are applied to the (closed) Toda lattices associated with the affine Lie algebras of type
$A_n^{(1)}$, $C_n^{(1)}$  and $A_{2n}^{(2)}$. 
In the $A_n^{(1)}$ case, corresponding to the classical periodic Toda lattices, Theorem \ref{thm:involution2} discloses a link between the PqN structure of these lattices 
and their bi-Hamiltonian representation in Flaschka coordinates.

This paper is organized as follows. 
Section \ref{sec:PN-PqN} is devoted to the definitions of PN and PqN manifolds, while in 
Section \ref{igor-connection} we recall how to deform a P(q)N manifold by means of a closed 2-form. 
In Section \ref{sec:involutivity} the two above mentioned involutivity theorems are proved, and they are 
exemplified in the cases of closed Toda lattices of type $A^{(1)}_n$ and $C^{(1)}_2$
in Section \ref{sec:Toda}.
In the final Section \ref{sec:Flaschka} we show that the well known bi-Hamiltonian structure of the 
$A_n^{(1)}$-Toda lattice (in Flaschka coordinates) can be interpreted as a projection of a PqN structure (more precisely, the $-\Omega$ deformation of the Das-Okubo PN structure), and that the recursion relations in the first setting come from the generalized ones in the PqN setting. The Appendices are devoted to the proof of Proposition \ref{prop:F-rel} and to the computational details of the $C_n^{(1)}$- and $A_{2n}^{(2)}$-Toda lattices, for generic $n$.


\section{Poisson quasi-Nijenhuis manifolds}
\label{sec:PN-PqN}

Let $N:T\M\to T\M$ be a $(1,1)$ tensor field on a manifold $\M$. We recall that its {\it Nijenhuis torsion\/} is defined as 
\begin{equation}
\label{tndef1}
T_N(X,Y)=[NX,NY]-N\left([NX,Y]+[X,NY]-N[X,Y]\right),
\end{equation}
and that, given a $p$-form $\alpha$, with $p\ge 1$, one can construct another $p$-form $i_N\alpha$ as \begin{equation}
\label{iNalpha}
i_N\alpha(X_1,\dots,X_p)=\sum_{i=1}^p \alpha(X_1,\dots,NX_i,\dots,X_p).
\end{equation}
If $\pi$ is a Poisson bivector on $\M$ and $\pi^\sharp:T^*\M\to T\M$ is defined by $\langle \beta,\pi^\sharp\alpha\rangle=\pi(\alpha,\beta)$, then 
$\pi$ and $N$ are said to be {\it compatible\/} \cite{MagriMorosiRagnisco85} if 
\begin{equation}
\label{N-P-compatible}
\begin{split}
&N\pi^\sharp=\pi^\sharp N^*,\,\,
\mbox{where $N^*:T^*\M\to T^*\M$ is the transpose of $N$;}\\
&L_{\pi^\sharp\alpha}(N) X-\pi^\sharp
L_{X}(N^*\alpha)+\pi^\sharp
L_{NX}\alpha=0,\,\,\mbox{for all 1-forms $\alpha$ and vector fields $X$.}
\end{split}
\end{equation}
In \cite{SX} a {\it Poisson quasi-Nijenhuis (PqN) manifold\/} was defined as a quadruple $(\M,\pi,N,\phi)$ such that:
\begin{itemize}
\item the Poisson bivector $\pi$ and the $(1,1)$ tensor field $N$  
are compatible;
\item the 3-forms $\phi$ and $i_N\phi$ are closed;
\item $T_N(X,Y)=\pi^\sharp\left(i_{X\wedge Y}\phi\right)$ for all vector fields $X$ and $Y$, where $i_{X\wedge Y}\phi$ is the 1-form defined as $\langle i_{X\wedge Y}\phi,Z\rangle=\phi(X,Y,Z)$.
\end{itemize}

If $\phi=0$, then the torsion of $N$ vanishes and $\M$ becomes a {\it Poisson-Nijenhuis manifold} (see \cite{KM} and references therein). In this case, the bivector field $\pi_N^{\phantom{\sharp}}$ defined by $\pi_N^\sharp=N\pi^\sharp$ is a Poisson tensor compatible with $\pi$, so that $\M$ is a bi-Hamiltonian manifold. 
Moreover, the functions
\begin{equation}
\label{tracce}
H_k=\frac1{2k}\Tr(N^k),\qquad k=1,2,\dots,
\end{equation}
satisfy $\d H_{k+1}=N^* \d H_{k}$, entailing
the so-called {\it Lenard-Magri relations\/} 
\begin{equation}
\label{LM-rel}
\pi^\sharp\, \d H_{k+1}=\pi_N^\sharp\, \d H_{k}
\end{equation}
and therefore their involutivity with respect to both Poisson brackets induced by $\pi$ and $\pi_N$. 
The involutivity (with respect to 
$\pi$) of the $H_k$ in the PqN case was discussed in \cite{FMOP2020} and will be further elaborated in Section \ref{sec:involutivity}.
\begin{rem} 
For a more general definition of PqN manifold, see \cite{BursztynDrummond2019}.
See also \cite{BursztynDrummondNetto2021}, where the PqN structures are recast in the 
framework of the Dirac-Nijenhuis ones. Another interesting generalization, given by the so called PqN manifolds with background, was considered in
\cite{Antunes2008,C-NdC-2010}.
\end{rem}


\section{Deformations of PqN manifolds}
\label{igor-connection}

In this section we remind a result concerning the deformations of a PqN structure. 
To do that, we need to recap a few well known notions. 

First of all we recall that given a tensor field $N:T\M\to T\M$, the usual Cartan differential can be modified as 
\begin{equation}
\label{eq:dNd}
\d_N=i_N\circ \d-\d\circ i_N,
\end{equation}
where $i_N$ is given by \eqref{iNalpha}.
We also remind that one can define a Lie bracket between 1-forms on a Poisson manifold $(\M,\pi)$ as
\begin{equation}
\label{eq:liealgpi}
[\alpha,\beta]_\pi=L_{\pi^\sharp\alpha}\beta-L_{\pi^\sharp\beta}\alpha-\d\langle\beta,\pi^\sharp\alpha\rangle,
\end{equation}
and that this Lie bracket 
can be uniquely extended to all forms on $\M$ in such a way that, if $\eta$ is a $q$-form and $\eta'$ is a $q'$-form, then 
$[\eta,\eta']_\pi$ is a $(q+q'-1)$-form and 
\begin{itemize}
\item[(K1)] $[\eta,\eta']_\pi=-(-1)^{(q-1)(q'-1)}[\eta',\eta]_\pi$; 
\item[(K2)] $[\alpha,f]_\pi=i_{\pi^\sharp\alpha}\,\d f=\langle \d f,\pi^\sharp\alpha\rangle$ for all $f\in C^\infty(\M)$ and for all 1-forms $\alpha$;
\item[(K3)] 
$[\eta,\cdot]_\pi$ 
is a derivation of degree $q-1$ of the wedge product, that is, 
for any differential form $\eta''$,
\begin{equation}
\label{deriv-koszul}
[\eta,\eta'\wedge\eta'']_\pi=[\eta,\eta']_\pi\wedge\eta''+(-1)^{(q-1)q'}\eta'\wedge[\eta,\eta'']_\pi.
\end{equation}
\end{itemize}
This extension is a {\it graded\/} Lie bracket, in the sense that (besides (K1)) the graded Jacobi identity holds,
\begin{equation}
\label{graded-jacobi}
(-1)^{(q_1-1)(q_3-1)}[\eta_1,[\eta_2,\eta_3]_\pi]_\pi+(-1)^{(q_2-1)(q_1-1)}[\eta_2,[\eta_3,\eta_1]_\pi]_\pi+(-1)^{(q_3-1)(q_2-1)}[\eta_3,[\eta_1,\eta_2]_\pi]_\pi=0,
\end{equation}
where $q_i$ is the degree of $\eta_i$.
It is sometimes called the Koszul bracket --- see, e.g., \cite{FiorenzaManetti2012} and references therein. We warn the reader that the Koszul bracket used in \cite{FMOP2020} is the opposite of the one used here, since a minus sign in (K2) was inserted. 

The following result has been proved in \cite{DMP2024}, generalizing that in \cite{FMP2023}, where the starting point is a PN manifold.
\begin{theorem}
\label{thm:gim}
Let $(\M,\pi,N,\phi)$ be a PqN manifold and let $\Omega$ be a closed 2-form. Define as usual $\Omega^\flat:T\M\to T^*\M$ as 
$\Omega^\flat(X)=i_X\Omega$. If 
\begin{equation}
\label{N-phi-hat}
\widehat N=N+\pi^\sharp\,\Omega^\flat\qquad\mbox{and}\qquad\widehat\phi=\phi+\d_N\Omega+\frac{1}{2}[\Omega,\Omega]_\pi,
\end{equation} 
then $(\M,\pi,\widehat N,\widehat\phi)$ is a PqN manifold. 
\end{theorem}

A preliminary version of this theorem was applied in \cite{FMP2023} to the classical closed Toda lattice --- more about this in Section \ref{sec:Toda}.

\section{Involutivity theorems for PqN manifolds}
\label{sec:involutivity}

In Section \ref{sec:PN-PqN} we recalled that the traces $H_k$ of the powers of $N$ are in involution in the PN case. This is not always true on a PqN manifold.
Some sufficient conditions for the involutivity of the $H_k$ were found in \cite{FMOP2020}. In this section we state Theorem \ref{thm:involution}, an improved version of that result. Moreover, using the notions of generalized Lenard-Magri chain, we prove an involutivity result for PqN manifolds that are deformations of PN manifold.

It is well known that, for any tensor field $N$ of type (1,1),   
\begin{equation}
\label{tndef2}
T_N(X,Y)=(L_{NX}N-NL_XN)Y,
\end{equation}
so that 
\begin{equation}
\label{ff1}
i_XT_N=L_{NX}N-N\, L_X N,
\end{equation}
where $i_XT_N$ is the $(1,1)$ tensor field 
defined as $(i_XT_N)(Y)=T_N(X,Y)$. 

Now let $(\M,\pi,N,\phi)$ be a PqN manifold. To study the involutivity of the functions $H_k=\frac1{2k}\Tr N^k$, it was noticed 
in \cite{FMOP2020} that, for $k\ge 1$ and for a generic vector field $X$ on $\M$, 
\begin{equation}
\label{req1}
\begin{split}
\brke{\d H_{k+1}}{X}&= L_X\left(\frac{1}{2(k+1)} \Tr(N^{k+1})\right)=\frac12\Tr\left((NL_X N)N^{k-1}\right)\\
&\buildrel(\ref{ff1})\over=\frac12\Tr\left(L_{NX}(N)N^{k-1}\right)-\frac12\Tr\left((i_XT_N)\, N^{k-1}\right)\\ 
&=L_{NX}\left(\frac1{2k} \Tr(N^{k})\right)-\frac12\Tr\left((i_XT_N)\, N^{k-1}\right)\\ 
&=\brke{\d H_{k}}{NX}-\frac12\Tr\left((i_XT_N)\, N^{k-1}\right)\\ 
&=\brke{N^*\d H_{k}}{X}-\frac12\Tr\left((i_XT_N)\, N^{k-1}\right).
\end{split}
\end{equation}
Hence the 
recursion 
relations 
\begin{equation}
\label{ff2}
N^* \d H_{k}=\d H_{k+1}+{\phi_{k-1}}
\end{equation}
were obtained, where 
\begin{equation}
\label{varphi}
\brke{\phi_k}{ X}= \frac12\Tr\left((i_XT_N)\, N^{k}\right)= \frac12\Tr\left(N^{k}\, (i_XT_N)\right),\qquad k\ge 0.
\end{equation}
Finally, the formula 
\begin{equation}
\label{recadd}
\parpo{H_k}{H_j}-\parpo{H_{k-1}}{H_{j+1}}=-\brke{ \phi_{j-1}}{\pi\, \d H_{k-1}}-\brke{\phi_{k-2}}{ \pi\, \d H_j},
\qquad k>j \ge 1,
\end{equation}
was proved, where $\{\cdot,\cdot\}$ is the Poisson bracket corresponding to $\pi$.
\begin{rem} 
Notice that:
\begin{itemize} 
\item the 1-forms $\phi_k$ and 
relation (\ref{ff2}) were used in \cite{Bogo96-180,Bogo96-182} for different purposes;
\item the 1-forms called $\phi_k$ in \cite{FMOP2020} are twice the ones in (\ref{varphi}), because in that paper we considered the functions 
$\frac1{k}\Tr N^k$ instead of the $H_k$.
\end{itemize}
\end{rem}
We are ready to state and prove 
\begin{theorem} 
\label{thm:involution}
Let $(\M,\pi,N,\phi)$ be a PqN manifold,  
and $H_k=\frac1{2k}\Tr({N}^k)$. Suppose that there exists a 2-form $\Omega$ such that:
\begin{itemize}
\item[$(a)$] $\phi=-2\,\d H_1\wedge\Omega$; 
\item[$(b)$] $\Omega(X_j,{Y_k})=0$ for all $j,k\ge 1$, where $Y_k=N^{k-1}X_1-X_k$ and $X_k=\pi^\sharp\,\d H_k$.
\end{itemize}
Then 
$\parpo{H_j}{H_k}=0$ for all $j,k\ge 1$. 
\end{theorem}
\begin{proof} Since on a PqN manifold, for all vector fields $X,Y$, 
\begin{equation}
\label{eq:pre-ThatN}
T_{N}(X,Y)=\pi^\sharp(i_{X\wedge Y}\phi),
\end{equation}
we have that
\begin{equation}
\label{eq:ThatN}
\begin{split}
T_{N}(X,Y)& 
\stackrel{(a)}{=}-2\pi^\sharp(i_Yi_X (\d H_1\wedge\Omega))
=-2\pi^\sharp(i_Y(\langle \d H_1,X\rangle\Omega-\d H_1\wedge i_X\Omega))\\
&=-2\pi^\sharp(\langle \d H_1,X\rangle i_Y\Omega-\langle \d H_1,Y\rangle i_X\Omega+(i_Yi_X\Omega) \d H_1)\\
&=-2\langle \d H_1,X\rangle(\pi^\sharp\,\Omega^\flat)(Y)+2\langle \d H_1,Y\rangle(\pi^\sharp\,\Omega^\flat)(X)-2\Omega(X,Y)X_1,
\end{split}
\end{equation}
so that 
\begin{equation}
\label{eq:5}
i_XT_{N}=-2\langle \d H_1,X\rangle\,\pi^\sharp\,\Omega^\flat
+2(\pi^\sharp\,\Omega^\flat)(X)\otimes \d H_1-2X_1\otimes \Omega^\flat X.
\end{equation}
Therefore 
\begin{equation}
\label{eq-pre-tr}
\begin{split}
\brke{\phi_k}{X} &\stackrel{\eqref{varphi}}{=}\frac12\Tr\left({N}^k(i_{X}T_{N})\right)
=
\Tr\left({N}^k\left(-\langle \d H_1,X\rangle\,\pi^\sharp\,\Omega^\flat
+(\pi^\sharp\,\Omega^\flat)  (X)\otimes \d H_1
-X_1\otimes \Omega^\flat X\right)\right)\\
&=-
\langle \d H_1,X\rangle\Tr\left({N}^k\pi^\sharp\,\Omega^\flat\right)
+
\Tr\left(({N}^k\pi^\sharp\,\Omega^\flat)(X)\otimes \d H_1\right)
-
\Tr\left(({N}^k X_1)\otimes \Omega^\flat X\right).
\end{split}
\end{equation}
Now observe that the last two terms in the last line of \eqref{eq-pre-tr} sum up to $-2\Omega(X,N^kX_1)$. In fact,
\begin{equation}
\Tr(X\otimes\alpha)=\brke{\alpha}{X}\label{eq:identi} 
\end{equation} 
entails that $-
\Tr\left(({N}^k X_1)\otimes \Omega^\flat X\right)=-\Omega(X,N^kX_1)$ and
\begin{equation}
\label{eqtr-b}
\begin{split}
\Tr\left(({N}^k\pi^\sharp\,
\Omega^\flat)(X)\otimes \d H_1\right)
&=\brke{\d H_1}{({N}^k\pi^\sharp\,
\Omega^\flat)(X)}
=\brke{\d H_1}{(\pi^\sharp ({N}^*)^k \Omega^\flat)(X)}\\
&=-\brke{(({N}^*)^k \Omega^\flat)(X)}{X_1}
=-\brke{\Omega^\flat(X)}{{N}^k X_1}
=-\Omega(X,{N}^k X_1).
\end{split}
\end{equation}
The previous remarks imply that, for every vector field $X$,
\begin{equation}
\label{eqtr}
\brke{\phi_k}{X}=
-\langle \d H_1,X\rangle\Tr\left({N}^k\pi^\sharp\,\Omega^\flat\right)-2\Omega(X,{N}^k X_1),
\end{equation}
meaning that
\begin{equation}
\label{eq-phi_k}
\phi_k=2\Omega^\flat({N}^k X_1)-\Tr\left({N}^k\pi^\sharp\,\Omega^\flat\right)\d H_1.
\end{equation}
Putting 
$X=X_j$ in (\ref{eqtr}), one has
\begin{equation}
\label{eqtr2}
\brke{\phi_k}{X_j}=-\parpo{H_1}{H_j}\Tr\left({N}^k\pi^\sharp\,\Omega^\flat\right)-2\Omega(X_j,{N}^k X_1).
\end{equation}
To prove that the traces $H_k$ of the powers of $N$ are in involution, we show by induction on $n$ 
that the following property holds:
\begin{equation*}
(P_n):\qquad\parpo{H_l}{H_m}=0\quad\mbox{for all pairs $(l,m)$ such that $l+m\le n$.}
\end{equation*}
This is obviously true for $n=2$. Let us show that $(P_n)$ implies $(P_{n+1})$. If $l+m=n+1$ and $l>m$, then 
\begin{equation}
\label{recadd2}
\begin{split}
\parpo{H_l}{H_m}&\stackrel{\eqref{recadd}}{=}\parpo{H_{l-1}}{H_{m+1}}-\brke{ \phi_{m-1}}{X_{l-1}}-\brke{\phi_{l-2}}{X_m}\\
&\stackrel{\eqref{eqtr2}}{=}\parpo{H_{l-1}}{H_{m+1}}+\parpo{H_1}{H_{l-1}}\Tr\left({N}^{m-1}\pi^\sharp\,\Omega^\flat\right)+2\Omega(X_{l-1},{N}^{m-1} X_1)\\
&\ \ \ +\parpo{H_1}{H_m}\Tr\left({N}^{l-2}\pi^\sharp\,\Omega^\flat\right)+2\Omega(X_m,{N}^{l-2} X_1).
\end{split}
\end{equation}
Since $\parpo{H_1}{H_{l-1}}=\parpo{H_1}{H_m}=0$ by the induction hypothesis, we obtain
\begin{equation}
\label{recadd3}
\parpo{H_l}{H_m}
=\parpo{H_{l-1}}{H_{m+1}}+2\Omega(X_{l-1},{N}^{m-1} X_1)+2\Omega(X_m,{N}^{l-2} X_1).
\end{equation}
Now, thanks to assumption (b), we can substitute ${N}^{i-1}X_1$ with $X_{i}$ in the 
last two terms, showing that their sum vanishes. Hence we obtain that 
\begin{equation}
\label{recadd4}
\parpo{H_l}{H_m}=\parpo{H_{l-1}}{H_{m+1}}.
\end{equation}
In the same way, we can show that $\parpo{H_{l-1}}{H_{m+1}}=\parpo{H_{l-2}}{H_{m+2}}$ and so on, until we obtain that 
\begin{equation}
\label{recadd5}
\parpo{H_l}{H_m}=\parpo{H_{m}}{H_{l}}=0,
\end{equation}
so that $(P_{n+1})$ holds.
\end{proof}  
\begin{rem}
\begin{itemize} 
\item[]
\item This result can be applied to the PqN structure of the classical closed Toda lattice found in \cite{FMOP2020} --- see that paper for the proof that the hypotheses of Theorem \ref{thm:involution} are satisfied. We will show in Section \ref{sec:Toda} and in Appendix A that it can be applied also to other Toda lattices.
\item It can be checked that condition $(a)$ does not hold for the 
PqN manifold associated with the 3-particle Calogero model, meaning that it is not necessary for the involutivity.
\end{itemize} 
\end{rem}

In the next section we will need the following modified version of 
Theorem \ref{thm:involution}, concerning the case where the PqN manifold is a deformation of a PN manifold.

\begin{theorem} 
\label{thm:involution2}
Let $(\M,\pi,N)$ be a PN manifold and $\Omega$ a closed 2-form on $\M$. If 
\begin{equation}
\label{N_pm}
N_{\pm}=N\pm\pi^\sharp\Omega^\flat,\qquad
\phi_\pm=\pm \d_N\Omega+\frac12[\Omega,\Omega]_\pi, 
\end{equation}
then we know from Theorem \ref{thm:gim} that $(\M,\pi,N_\pm,\phi_\pm)$ are PqN manifolds. Define $H_k^+=\frac1{2k}\Tr({N}_+^k)$ and suppose that:
\begin{enumerate}
\item[$(a')$] $\phi_+=-2\,\d H^+_1\wedge\Omega$; 
\item[$(b')$] $\Omega^\flat({Y_k^+})=0$ for all $k\ge 1$, where $Y_k^+=N_+^{k-1}X^+_1-X^+_k$ and $X^+_k=\pi^\sharp\,\d H^+_k$.
\end{enumerate}
Then:
\begin{itemize}
\item[i)] the functions $H_k^+$ form a generalized Lenard-Magri chain 
with respect to $N_-$, in the sense that
\begin{equation}
\label{gen-LM-chain}
N_-^* \d H_k^+ = \d H_{k+1}^+ + f_k \d H_1^+,
\end{equation}
where $f_k=-\Tr\left({N}_+^{k-1}\pi^\sharp\,\Omega^\flat\right)$;
\item[ii)] 
$\parpo{H^+_j}{H^+_k}=0$ for all $j,k\ge 1$.
\end{itemize}
\end{theorem}
\noindent
\begin{proof}
{\ }\par\noindent
i) Let us consider the PqN manifold $(\M,\pi,N_+,\phi_+)$, and the corresponding relations (\ref{ff2}) and (\ref{eq-phi_k}), that is,
\begin{equation}
\label{ff2+eq-phi_k+}
N_+^* \d H^+_{k}=\d H^+_{k+1}+{\phi^+_{k-1}},\quad\text{where } \quad \phi^+_k=2\Omega^\flat({N}_+^k X_1^+)-\Tr\left({N}_+^k\pi^\sharp\,\Omega^\flat\right)\d H^+_1.
\end{equation}
Thanks to assumption $(b')$, we have that 
\begin{equation}
N_+^* \d H^+_{k}=\d H^+_{k+1}+2\Omega^\flat (X_k^+) -\Tr\left({N}_+^{k-1}\pi^\sharp\,\Omega^\flat\right)\d H^+_1.
\end{equation}
Hence the thesis follows from $N^*_+=N^*_- +2\Omega^\flat\pi^\sharp$.
\par\noindent
ii) It is clear that the ${H^+_k}$ are in involution as a consequence of Theorem \ref{thm:involution}. However, we think it is worthwhile to show that their involutivity follows from relations (\ref{gen-LM-chain}). This can be deduced by Proposition 5.1 in \cite{Tondo1995}, showing that the ${H^+_k}$ form a Nijenhuis chain, in the terminology of \cite{FMT2000} --- see also \cite{Magri2003} for the related notion of Lenard chain.
For the sake of consistency, 
we give a direct proof similar to that of Theorem \ref{thm:involution}, 
showing by induction on $n$ that the following property holds:
\begin{equation*}
(P_n):\qquad\parpo{H^+_l}{H^+_m}=0\quad\mbox{for all pairs $(l,m)$ such that $l+m\le n$.}
\end{equation*}
Since $(P_2)$ is trivial, we are left with showing that $(P_n)$ implies $(P_{n+1})$. If $l+m=n+1$ and $l<m$, then 
\begin{equation*}
\begin{aligned}
 \parpo{H^+_l}{H^+_m}
 &=\langle \d H^+_l,\pi^\sharp \d H^+_m\rangle
 =\langle \d H^+_l,\pi^\sharp \left(N^*_-\d H^+_{m-1}-f_{m-1}\d H^+_1\right)\rangle\\
 &=\langle \d H^+_l,N_-\pi^\sharp \d H^+_{m-1}\rangle-f_{m-1}\parpo{H^+_l}{H^+_1}
 =\langle N^*_-\d H^+_l,\pi^\sharp \d H^+_{m-1}\rangle,
\end{aligned}
\end{equation*}
since $l+1<m+1\le m+l=n+1$ implies that $\parpo{H^+_l}{H^+_1}=0$. Therefore
\begin{equation*}
\begin{aligned}
 \parpo{H^+_l}{H^+_m}
 &=\langle N^*_-\d H^+_l,\pi^\sharp \d H^+_{m-1}\rangle
 =\langle \d H^+_{l+1}+f_l \d H^+_1,\pi^\sharp \d H^+_{m-1}\rangle
 =\parpo{H^+_{l+1}}{H^+_{m-1}}+f_l \parpo{H^+_1}{H^+_{m-1}}\\
 &=\parpo{H^+_{l+1}}{H^+_{m-1}},
\end{aligned}
\end{equation*}
since $1+(m-1)=m<m+l=n+1$ implies that $\parpo{H^+_1}{H^+_{m-1}}=0$. 
In the same way, we can show that $\parpo{H^+_{l+1}}{H^+_{m-1}}=\parpo{H^+_{l+2}}{H^+_{m-2}}$ and so on, until we obtain that 
\begin{equation*}
\label{recadd5+}
\parpo{H^+_l}{H^+_m}=\parpo{H^+_{m}}{H^+_{l}}=0,
\end{equation*}
so that $(P_{n+1})$ holds.
\end{proof}
\begin{rem}
\begin{itemize}
\item[]
\item The factor -2 in assumptions $(a)$ and $(a')$ of Theorem \ref{thm:involution} and, respectively, Theorem \ref{thm:involution2}, was chosen to simplify the application of these results to the Toda cases.
\item The hypotheses of Theorem \ref{thm:involution2} are stronger than those of Theorem \ref{thm:involution}. However, the 
recursion relations (\ref{gen-LM-chain}) are more transparent then the ones in (\ref{ff2+eq-phi_k+}).
\end{itemize}
\end{rem}


\section{PqN structures for $A^{(1)}_n$- and $C^{(1)}_2$-type Toda lattices}
\label{sec:Toda}

In this section we apply the general setting of the previous one to Toda lattices associated with the affine Lie algebras of type $A_{n}^{(1)}$ and $C_{2}^{(1)}$. We refer to \cite{RSTS} for an introductory survey to 
generalized Toda systems within the Lie-algebraic framework.
Some of the results were already obtained in \cite{FMOP2020} for the classical closed Toda lattice, i.e., the one 
associated with the affine Lie algebra $A_n^{(1)}$. However, in that paper a preliminary 
version of Theorem \ref{thm:involution} was used.

\subsection{The $A_n^{(1)}$ case} The starting point is the PN manifold $(\mathbb R^{2n},\pi,N)$ introduced in \cite{DO} to describe the \emph{open} Toda lattice, i.e., the $A_n$ case. 
The Poisson tensor is the canonical one in the coordinates $(q_i,p_i)$, that is, $\pi=\sum_{i=1}^n \partial_{p_i}\wedge\partial_{q_i}$,
while $N$ is the (torsion free) 
tensor field 
\begin{equation}
\label{N-openToda}
\begin{split}
N&=\sum_{i=1}^{n}p_i\left(\partial_{q_i}\otimes \d q_i+\partial_{p_i}\otimes \d p_i\right)+
\sum_{i<j}\left(\partial_{q_i}\otimes \d p_j-\partial_{q_j}\otimes \d p_i\right)\\
&+\sum_{i=1}^{n-1}{\rm e}^{q_i-q_{i+1}}\left(\partial_{p_{i+1}}\otimes \d q_i-\partial_{p_i}\otimes \d q_{i+1}\right).
\end{split}
\end{equation}
We refer to  Section \ref{sec:Flaschka} for the corresponding matrix expressions in the 4-particle case, from which the general form is easily guessed.
The functions $H_k=\frac1{2k}\Tr(N^k)$ are 
integrals of motion of the open Toda lattice. For example,
\begin{equation}
\label{traces-openToda}
H_1=\sum_{i=1}^{n}p_i\qquad\mbox{and}\qquad H_2=\frac12\sum_{i=1}^{n}p_i^2+\sum_{i=1}^{n-1}{\rm e}^{q_i-q_{i+1}}
\end{equation}
are respectively the total momentum and the energy.

Now let us consider the closed 2-form $\Omega=\mathrm{e}^{q_n-q_1}\d q_n\wedge \d q_1$ on $\RR^{2n}$. Then, according to Theorem 
\ref{thm:gim}, we can deform the PN manifold $(\RR^{2n},\pi,N)$ to obtain two different PqN manifolds, namely, 
$(\RR^{2n},\pi,N_+,\phi_+)$ and $(\RR^{2n},\pi,N_-,\phi_-)$, where $N_\pm$ and $\phi_\pm$ are given by (\ref{N_pm}).
If $H^\pm_k=\frac1{2k}\Tr(N_\pm^k)$, then $H^+_1=H^-_1=H_1$ and 
\begin{equation}
\label{energy-closed-Toda}
H^\pm_2=\frac12\sum_{i=1}^{n}p_i^2+\sum_{i=1}^{n-1}{\rm e}^{q_i-q_{i+1}}\pm {\rm e}^{q_n-q_1},
\end{equation}
so that $H^+_2$ is the energy of the classical closed Toda lattice, while $H^-_2$ describes the case where the interaction between 
the first and the last particle is repulsive. 
It was shown in Theorem 7 of \cite{FMOP2020} (see Remark \ref{rem:phi-pm} below) that:
\begin{enumerate}
\item $\phi_+=-2\,\d H_1^+\wedge\Omega$; 
\item $\Omega^\flat({Y_k^+})=0$ for all $k\ge 1$, where $Y_k^+=N_+^{k-1}X_1^+ -X^+_k$ and $X^+_k=\pi^\sharp\,\d H^+_k$.
\end{enumerate}
Simply using $-\Omega$ instead of $\Omega$, we have that:
\begin{enumerate}
\item[3.] $\phi_-=2\,\d H_1^-\wedge\Omega$; 
\item[4.] $\Omega^\flat({Y_k^-})=0$ for all $k\ge 1$, where $Y_k^-=N_-^{k-1}X^-_1-X^-_k$ and $X^-_k=\pi^\sharp\,\d H^-_k$.
\end{enumerate}
Hence we can apply Theorems \ref{thm:involution} and \ref{thm:involution2} to conclude that:
\begin{itemize}
\item[i)] 
$\parpo{H^+_j}{H^+_k}=\parpo{H^-_j}{H^-_k}=0$ for all $j,k\ge 1$;
\item[ii)] the functions $H_k^\pm$ form a generalized Lenard-Magri chain 
with respect to $N_\mp$, in the sense that
\begin{equation}
\label{gen-LM-chain-pm}
N_-^* \d H_k^+ = \d H_{k+1}^+ + f^+_k \d H_1^+ \qquad\mbox{and}\qquad
N_+^* \d H_k^- = \d H_{k+1}^- + f^-_k \d H_1^-,
\end{equation}
where $f^\pm_k=\mp\Tr\left({N}_\pm^{k-1}\pi^\sharp\,\Omega^\flat\right)$.
\end{itemize}

\begin{rem}
It is easily checked that $ f^+_1= f^-_1=0$, so that the first recursion relations are 
\begin{equation}
\label{gen-LM-chain-pm-first}
N_-^* \d H_1^+ = \d H_{2}^+ \qquad\mbox{and}\qquad
N_+^* \d H_1^- = \d H_{2}^- .
\end{equation}
\end{rem}

\begin{rem}
\label{rem:phi-pm}
In \cite{FMOP2020} it was noticed that $[\Omega,\Omega]_\pi=0$, entailing that $\phi_\pm=\pm\d_N\Omega$, according to (\ref{N_pm}).
Hence the equality $\phi_+=-2\,\d H_1^+\wedge\Omega$ amounts to $\d_N\Omega=-2\,\d H_1^+\wedge\Omega$, which was proved in 
\cite{FMOP2020} (even though the minus sign is missing due to a misprint).
\end{rem}

Theorem \ref{thm:involution} can be applied also to the Toda lattices associated with the affine Lie algebras $C_{n}^{(1)}$ and $A_{2n}^{(2)}$, as shown in Appendix A in full details. 
In this section, for simplicity, we consider the $C_2^{(1)}$ case only.

We remark that, given a bivector $\pi$ and a coordinate system 
$(x_1,\dots, x_n)$ on a manifold $\M$, we have that $X=\pi^\sharp\alpha$ if and only if $X^j=\pi^{ij}\alpha_i$ (summation over repeated indices is understood). 
Since we prefer to use column rather than row vectors, whenever we write $\pi
=A$ and $A$ is a matrix, we mean that the $(i,j)$ entry of $A$ is $\pi^{ji}$. 
For the same reason, when $N$ is a (1,1) tensor field and we write $N=A$, 
we mean that the $(i,j)$ entry of $A$ is $N_j^i$. 
Moreover, in this section, in the next one, and in Appendix B we use (almost everywhere) the same notation for a bivector $P$ and its associated map $P^\sharp$.
\subsection{The $C^{(1)}_2$ case}\label{ss:C2} 
The open Toda lattice corresponding to the Lie algebra $C_2$ 
is the 2-particle 
system whose Hamiltonian is
\[
H_{C_2}=\frac{1}{2}(p_1^2+p_2^2)+e^{q_1-q_2}+e^{2q_2}.
\]
The Poisson bivector, see \cite{daCosta-Damianou},
\begin{equation}
{\pi'}
=\left( \begin {array}{cc|cc} 
    0 & 2p_2 & -p_1^2-2e^{q_1-q_2} & e^{q_1-q_2}-4e^{2q_2}\\
    -2p_2 & 0 & -e^{q_1-q_2} & -p_2^2-2e^{2q_2}\\
    \hline
    p_1^2+2e^{q_1-q_2} & e^{q_1-q_2} & 0 & e^{q_1-q_2}(p_1+p_2)\\
    4e^{2q_2}-e^{q_1-q_2} & p_2^2+2e^{2q_2} & -e^{q_1-q_2}(p_1+p_2) & 0 
\end{array}\right)
\label{eq: C2 Poisson tensor},    
\end{equation}
together with the canonical one, 
\begin{equation*} 
\pi
=\left( \begin {array}{cc|cc} 
0&0&\,1\, &0\\ 
0&0&0&\,1\\ 
\hline
-1&0&0&0\\ 
0&-1&0&0
\end{array}\right),
\end{equation*}
provides a bi-Hamiltonian formulation for the system. Hence, the phase space $\RR^4$ 
is endowed with the PN structure $(\pi,N)$, where 
\begin{equation}
{N}=
{\pi'}\,
\pi
^{-1}=
\left( \begin {array}{cc|cc}
-p_1^2-2e^{q_1-q_2} & e^{q_1-q_2}-4e^{2q_2} & 0 & -2p_2\\
-e^{q_1-q_2} & -p_2^2-2e^{2q_2} & 2p_2 & 0\\
\hline
0 & e^{q_1-q_2}(p_1+p_2) & -p_1^2-2e^{q_1-q_2} & -e^{q_1-q_2}\\
-e^{q_1-q_2}(p_1+p_2) & 0 & e^{q_1-q_2}-4e^{2q_2} & -p_2^2-2e^{2q_2}
\end{array}\right).
\label{eq: open C2 (1,1) tensor}
\end{equation}
By means of the closed 2-form $\Omega_1=\d(e^{-2q_1}\d p_1)=-2e^{-2q_1}\d q_1\wedge \d p_1$, we can use Theorem 
\ref{thm:gim} to deform the PN structure $(\pi,N)$ and we obtain the PqN structure $(\pi,\widehat{N},\widehat{\phi})$, where \begin{equation}
\widehat{N}=N+\pi^{\sharp}\Omega_1^{\flat}=\left( \begin {array}{cc|cc}
-p_1^2-2e^{q_1-q_2}-2e^{-2q_1} & e^{q_1-q_2}-4e^{2q_2} & 0 & -2p_2\\
-e^{q_1-q_2} & -p_2^2-2e^{2q_2} & 2p_2 & 0\\
\hline
0 & e^{q_1-q_2}(p_1+p_2) & -p_1^2-2e^{q_1-q_2}-2e^{-2q_1} & -e^{q_1-q_2}\\
-e^{q_1-q_2}(p_1+p_2) & 0 & e^{q_1-q_2}-4e^{2q_2} & -p_2^2-2e^{2q_2}
\end{array}\right)
\label{eq: periodic C2 (1,1) tensor}
\end{equation}
and $\widehat{\phi}=\d_{N}\Omega_1+\frac{1}{2}\left[\Omega_1,\Omega_1\right]_\pi$. Notice that 
\[
H_1=
\frac{1}{2}\Tr(\widehat{N})=-\left(p_1^2+p_2^2+2e^{q_1-q_2}+2e^{2q_2}+2e^{-2q_1}\right)=-2H_{C_2^{(1)}},
\]
where
\[
H_{C_2^{(1)}}=\frac{1}{2}\left(p_1^2+p_2^2\right)+e^{q_1-q_2}+e^{2q_2}+e^{-2q_1}
\]
is the Hamiltonian of the 2-particle $C_2^{(1)}$-Toda lattice. Notice also that
\begin{equation}
\widehat{\phi}=\d_{N}\Omega_1+\frac{1}{2}{\left[\Omega_1,\Omega_1\right]_\pi}
=\d_{N}\Omega_1,\label{eq: 3-form C2}
\end{equation}
since the vanishing of the Koszul bracket $\left[\Omega_1,\Omega_1\right]_\pi$ follows from easy computations involving the facts that $\Omega_1$ is exact and $\d$ is a derivation of the Koszul bracket.

Moreover,  using $\d_{N}\circ \d=-\d\circ \d_{N}$, $\d_{N}f=N^*(\d f)$ for smooth functions $f$, 
and $\d_{N}(f \d x_i)=(\d_{N}f)\wedge \d x_i+f(\d_{N}(\d x_i))$, one can see that
\begin{equation}
\widehat{\phi}=\left(-4e^{-2q_1}\d q_1\wedge N^*(\d q_1)+2e^{-2q_1}\d(N^*(\d q_1))\right)\wedge \d p_1-2e^{-2q_1}\d q_1\wedge \d(N^*(\d p_1)).\label{eq: phi_expression}
\end{equation}
Using the equations
\begin{equation}\label{eq: differentials}
\begin{aligned}
N^{*}(\d q_1)&=-(p_1^2+2e^{q_1-q_2})\d q_1+(e^{q_1-q_2}-4e^{2q_2})\d q_2-2p_2\d p_2\\
N^*(\d p_1)&=e^{q_1-q_2}(p_1+p_2)\d q_2-(p_1^2+2e^{q_1-q_2})\d p_1-e^{q_1-q_2}\d p_2\\
\d(N^{*}(\d q_1))&=-2p_1\d p_1\wedge \d q_1-e^{q_1-q_2}\d q_1\wedge \d q_2\\
\d(N^*(\d p_1))&=e^{q_1-q_2}(p_1+p_2)\d q_1\wedge \d q_2+e^{q_1-q_2}\d q_2\wedge \d p_1-2e^{q_1-q_2}\d q_1\wedge \d p_1-e^{q_1-q_2}\d q_1\wedge \d p_2
\end{aligned}
\end{equation}
one finds that 
\begin{align*}
\widehat{\phi}&=-2e^{-2q_1}\left[\left(4e^{q_1-q_2}-8e^{2q_2}\right)\d q_1\wedge \d q_2\wedge \d p_1+4p_2\d q_1\wedge \d p_1\wedge \d p_2\right]\\
&=-2\,\d H_1\wedge\Omega_1,
\end{align*}
so that the decomposition of $\widehat{\phi}$ as in condition $(a)$ of Theorem \ref{thm:involution} is verified. 
Condition $(b)$ 
is checked for the general case $C_n^{(1)}$ in Appendix A. This shows that 
the phase space of the 2-particle $C_2^{(1)}$-Toda lattice has a formulation in terms of involutive PqN manifolds.
\begin{rem} Actually, in Appendix A we show that the stronger condition $(b')$, appearing in Theorem \ref{thm:involution2}, is satisfied. Hence this theorem too can applied to the $C_n^{(1)}$ case. 
Since $\left[\Omega_1,\Omega_1\right]_\pi=0$, it turns out that 
$\phi_\pm=\pm \d_N\Omega_1$, as in the $A_n^{(1)}$ case.
\end{rem}

\section{Flaschka coordinates and reduction in the $A_n^{(1)}$ case}
\label{sec:Flaschka}

We consider the PqN manifold $(\RR^{2n},\pi,N_-,\phi_-)$ introduced in the previous section. 
In this section we will show that the following claim holds true:\\
\indent\emph{The above mentioned PqN structure, describing the closed Toda lattice in 
physical variables, reduces to the bi-Hamiltonian structure of the closed Toda lattice in Flaschka coordinates, 
see \eqref{Flaschka-map} below. 
Under this reduction, the generalized Lenard-Magri chain \eqref{gen-LM-chain} goes to the standard one, see \eqref{gen-LM-chain-3} below.}

To avoid unnecessary complications, here we shall limit to display the concrete expressions for the $4$-particle  case. The generic $n$-particle case, with full proofs of the assertions below, is contained in Appendix B.

According to the conventions and settings of the previous section, the canonical Poisson structure on $\RR^{8}$ is represented as
\begin{equation*} 
\pi
=\left( \begin {array}{cccc|cccc} 
0&0&0&0&\,1\, &0&0&0\\ 
0&0&0&0&0&\,1\,&0&0\\ 
0&0&0&0&0&0&\,1\,&0\\ 
0&0&0&0&0&0&0&\,1\,\\ 
\hline
-1&0&0&0&0&0&0&0\\ 
0&-1&0&0&0&0&0&0\\ 
0&0&-1&0&0&0&0&0\\ 
0&0&0&-1&0&0&0&0\\ 
\end{array}\right),
\end{equation*}
while

\begin{equation*}
N=
\left(\begin {array}{cccc|cccc} 
p_{{1}}&0&0&0&0&1&1&1\\ 
0&p_{{2}}&0&0&-1&0&1&1\\ 
0&0&p_{{3}}&0&-1&-1&0&1\\ 
0&0&0&p_{{4}}&-1&-1&-1&0\\ 
\hline
0&-{{\rm e}^{q_{{1}}-q_{{2}}}}&0& 0 &p_{{1}}&0&0&0\\ 
{{\rm e}^{q_{{1}}-q_{{2}}}}&0&-{{\rm e}^{q_{{2}}-q_{{3}}}}&0&0&p_{{2}}&0&0\\ 
0&{{\rm e}^{q_{{2}}-q_{{3}}}}&0&-{{\rm e}^{q_{{3}}-q_{{4}}}}&0&0&p_{{3}}&0\\ 
0&0&{{\rm e}^{q_{{3}}-q_{{4}}}}&0&0&0&0&p_{{4}}
\end {array}\right)
\end{equation*}
and
\begin{equation*}
N_-=
\left(\begin {array}{cccc|cccc} 
p_{{1}}&0&0&0&0&1&1&1\\ 
0&p_{{2}}&0&0&-1&0&1&1\\ 
0&0&p_{{3}}&0&-1&-1&0&1\\ 
0&0&0&p_{{4}}&-1&-1&-1&0\\ 
\hline
0&-{{\rm e}^{q_{{1}}-q_{{2}}}}&0& {{\rm e}^{q_{{4}}-q_{{1}}}}&p_{{1}}&0&0&0\\ 
{{\rm e}^{q_{{1}}-q_{{2}}}}&0&-{{\rm e}^{q_{{2}}-q_{{3}}}}&0&0&p_{{2}}&0&0\\ 
0&{{\rm e}^{q_{{2}}-q_{{3}}}}&0&-{{\rm e}^{q_{{3}}-q_{{4}}}}&0&0&p_{{3}}&0\\ 
-{{\rm e}^{q_{{4}}-q_{{1}}}}&0&{{\rm e}^{q_{{3}}-q_{{4}}}}&0&0&0&0&p_{{4}}
\end {array}\right).
\end{equation*}
The bivector field 
defined by $\pi_{N_-}^{\phantom{\sharp}}
=N_-\pi
$ is thus
\begin{equation}
\label{piN-Toda}
\pi_{N_-}^{\phantom{\sharp}}
=\left( \begin {array}{cccc|cccc} 
0&-1&-1&-1&p_{{1}}&0&0&0\\ 
1&0&-1&-1&0&p_{{2}}&0&0\\ 
1&1&0&-1&0&0&p_{{3}}&0\\ 
1&1&1&0&0&0&0&p_{{4}}\\ 
\hline
-p_{{1}}&0&0&0&0&-{{\rm e}^{q_{{1}}-q_{{2}}}}&0&{{\rm e}^{q_{{4}}-q_{{1}}}}\\ 
0&-p_{{2}}&0&0&{{\rm e}^{q_{{1}}-q_{{2}}}}&0&-{{\rm e}^{q_{{2}}-q_{{3}}}}&0\\ 
0&0&-p_{{3}}&0&0&{{\rm e}^{q_{{2}}-q_{{3}}}}&0&-{{\rm e}^{q_{{3}}-q_{{4}}}}\\ 
0&0&0&-p_{{4}}&-{{\rm e}^{q_{{4}}-q_{{1}}}}&0&{{\rm e}^{q_{{3}}-q_{{4}}}}&0
\end {array} \right).
\end{equation}

Let us consider the Flaschka map
$F:\RR^{2n}\to \RR^{2n}$ given by $$F(q_1,\dots,q_n,p_1,\dots,p_n)=(a_1,\dots,a_n,b_1,\dots,b_n),$$ where
\begin{equation}
\label{Flaschka-map}
a_i=-{\rm e}^{q_i-q_{i+1}},\qquad b_i=p_i,\qquad\mbox{with $q_{n+1}=q_1$.}
\end{equation}
It is well know (see, e.g., \cite{Nunes-Marle} for the open case) that the canonical Poisson tensor $\pi$ is $F$-related 
to the Poisson bivector 
$$
P_0=\left(
\begin{array}{cccc|cccc}
 0 & 0 & 0 & 0 & a_1 & -a_1 & 0 & 0 \\
 0 & 0 & 0 & 0 & 0 & a_2 & -a_2 & 0 \\
 0 & 0 & 0 & 0 & 0 & 0 & a_3 & -a_3 \\
 0 & 0 & 0 & 0 & -a_4 & 0 & 0 & a_4 \\
 \hline
 -a_1 & 0 & 0 & a_4 & 0 & 0 & 0 & 0 \\
 a_1 & -a_2 & 0 & 0 & 0 & 0 & 0 & 0 \\
 0 & a_2 & -a_3 & 0 & 0 & 0 & 0 & 0 \\
 0 & 0 & a_3 & -a_4 & 0 & 0 & 0 & 0 \\
\end{array}
\right)
$$
in the sense that $P_0
=F_*\pi
F^*$. In Appendix B we prove for the general $n$-particle case
\begin{prop}
\label{prop:F-rel}
The bivector $\pi_{N_-}^{\phantom{\sharp}}$ is $F$-related to the Poisson bivector 
\begin{equation}
\label{P1-Toda}
P_1=\left(
\begin{array}{cccc|cccc}
 0 & -a_1 a_2 & 0 & a_1 a_4 & a_1 b_1
   & -a_1 b_2 & 0 & 0 \\
 a_1 a_2 & 0 & -a_2 a_3 & 0 & 0 & a_2
   b_2 & -a_2 b_3 & 0 \\
 0 & a_2 a_3 & 0 & -a_3 a_4 & 0 & 0 &
   a_3 b_3 & -a_3 b_4 \\
 -a_1 a_4 & 0 & a_3 a_4 & 0 & -a_4 b_1
   & 0 & 0 & a_4 b_4 \\
   \hline
 -a_1 b_1 & 0 & 0 & a_4 b_1 & 0 & a_1
   & 0 & -a_4 \\
 a_1 b_2 & -a_2 b_2 & 0 & 0 & -a_1 & 0
   & a_2 & 0 \\
 0 & a_2 b_3 & -a_3 b_3 & 0 & 0 & -a_2
   & 0 & a_3 \\
 0 & 0 & a_3 b_4 & -a_4 b_4 & a_4 & 0
   & -a_3 & 0 \\
\end{array}
\right).
\end{equation}
\end{prop}
\begin{rem}
As we shall see in Appendix B, the key property for this to hold is that the tangent $F_*$  to the Flaschka map is given by
\begin{equation}\label{Fstar4}
F_*=\left(
\begin{array}{cccc|cccc}
 a_1 & -a_1 & 0 & 0 & 0 & 0 & 0 & 0 \\
 0 & a_2 & -a_2 & 0 & 0 & 0 & 0 & 0 \\
 0 & 0 & a_3 & -a_3 & 0 & 0 & 0 & 0 \\
 -a_4 & 0 & 0 & a_4 & 0 & 0 & 0 & 0 \\
 \hline
 0 & 0 & 0 & 0 & 1 & 0 & 0 & 0 \\
 0 & 0 & 0 & 0 & 0 & 1 & 0 & 0 \\
 0 & 0 & 0 & 0 & 0 & 0 & 1 & 0 \\
 0 & 0 & 0 & 0 & 0 & 0 & 0 & 1 \\
\end{array}\right).
\end{equation}
\end{rem}
For the bi-Hamiltonian pair $(P_0,P_1)$ of the closed Toda lattice in Flaschka coordinates and its interpretation in the framework of linear and quadratic Poisson structures associated to an $r$-matrix, we refer to \cite{MorosiPizzocchero96} and references therein (see also \cite{FMP2001} for the relations with the separation of variables). Notice that $P_1$ is Poisson even though 
$\pi_{N_-}^{\phantom{\sharp}}$ is not. This can be explained as follows.

First of all, we notice that the image of the map $F$ is 
the submanifold 
$$
\widetilde\M=\{(a,b)\in \RR^{2n}\mid C(a,b)=(-1)^n, a_i<0\},
$$
where $C(a,b)=\prod_{i=1}^n a_i$, and that $F: \RR^{2n}\to\widetilde\M$ is the projection 
along the integral curves of the vector field $X_1=X^+_1=X^-_1=\sum_{i=1}^{n}\partial_{q_i}$. 
Since the pair $(\pi,\pi_{N_-}^{\phantom{\sharp}})$ is $F$-related to the pair $(P_0,P_1)$, it can be projected on the restriction of $(P_0,P_1)$ to $\widetilde\M$. 
Such restriction exists since $C$ is a Casimir function both for $P_0$ and $P_1$.
Now, let us first recall from \cite{SX} that  
$$
[\pi,\pi_{N_-}^{\phantom{\sharp}}]=0,\qquad [\pi_{N_-}^{\phantom{\sharp}},\pi_{N_-}^{\phantom{\sharp}}]=2\pi^\sharp\phi_-,
$$
where $[\cdot,\cdot]$ denotes the Schouten bracket between multi-vectors.
The first identity implies that $[\widetilde{P}_0,\widetilde{P}_1]=0$, where $\widetilde{P}_0$ and $\widetilde{P}_1$ are the restrictions of $P_0$ and $P_1$ on $\widetilde\M$. From the second identity it follows that 
$$
[\pi_{N_-}^{\phantom{\sharp}},\pi_{N_-}^{\phantom{\sharp}}]=2\pi^\sharp\phi_-=4\pi^\sharp(\d H_1^-\wedge\Omega)
=4\pi^\sharp(\d H_1^-)\wedge\pi^\sharp(\Omega)=4X_1\wedge\pi^\sharp(\Omega).
$$
Since the projection of $X_1$ on $\widetilde\M$ vanishes, we have that $[\widetilde{P}_1,\widetilde{P}_1]=0$.
Hence we can conclude that the PqN structure $(\RR^{2n},\pi,N_-,\phi_-)$ --- 
more precisely, the triple $(\RR^{2n},\pi,\pi_{N_-}^{\phantom{\sharp}})$ --- projects on the restriction 
$(\widetilde\M,\widetilde{P}_0,\widetilde{P}_1)$ of the bi-Hamiltonian manifold $(\RR^{2n},P_0,P_1)$.

As far as the integrals of motion $H^k_+$ of the closed Toda lattice are concerned, we first remark that they pass to the 
quotient $\widetilde\M$, in the sense that there exist functions $\widetilde{H}^k_+:\widetilde\M\to\RR$ such that 
${H}^k_+=\widetilde{H}^k_+\circ F$. This can be seen as a consequence of $X_1({H}^k_+)=0$ or, more directly, from the fact that the functions 
${H}^k_+$ depend only on the differences $q_i-q_{j}$. The generalized Lenard-Magri relations \eqref{gen-LM-chain} entail that 
\begin{equation}
\label{gen-LM-chain-2}
\pi_{N_-}^\sharp \d H_k^+ = \pi^\sharp \,\d H_{k+1}^+ + f_k X_1,\qquad k\ge 1,
\end{equation}
so that on $\widetilde\M$ 
we have the usual relations
\begin{equation}
\label{gen-LM-chain-3}
 \widetilde{P}_1^\sharp\,\d \widetilde{H}_k^+ = \widetilde{P}_0^\sharp\,\d \widetilde{H}_{k+1}^+,\qquad k\ge 1.
\end{equation}
One can easily check that $\widetilde{P}_0^\sharp\,\d \widetilde{H}_{1}^+=0$. Actually, 
${P}_0^\sharp\,\d {H}_{1}^+=0$, and relations (\ref{gen-LM-chain-3}) hold without tildes too.

\section*{Appendix A: Toda lattices related to $C_n^{(1)}$ and $A_{2n}^{(2)}$}
\label{app:Cn}
The Hamiltonians of the open orthogonal Toda systems 
are 
\begin{equation}\label{eq: Hamiltonian B & C}
\frac{1}{2}\sum_{i=1}^np_i^2+\sum_{i=1}^{n-1}e^{q_{i-1}-q_i}+e^{mq_n},   
\end{equation}
where $m=1$ (respectively, $m=2$) corresponds to the Lie algebras $B_n$ (respectively, $C_n$). 
In \cite{daCosta-Damianou} Poisson bivectors $
{\pi'}$ were introduced for the open orthogonal Toda systems, compatible with the canonical Poisson bivector 
$\pi$ and providing a bi-Hamiltonian formulation for the above mentioned Hamiltonian systems. 
The non-zero brackets 
of $
{\pi'}$ are, for $n\ge 3$, 
\begin{equation}\label{eq: B_n & C_n brackets}
\begin{aligned}
\lbrace q_i,q_{i-1}\rbrace'&=\lbrace q_i,q_{i-2}\rbrace'=\cdots=\lbrace q_i,q_1\rbrace'=2p_i,\quad i=2,\ldots,n,\\
\lbrace p_i,q_{i-2}\rbrace'&=\lbrace p_i,q_{i-3}\rbrace'=\cdots=\lbrace p_i,q_1\rbrace'=2\left(e^{q_{i-1}-q_i}-e^{q_i-q_{i+1}}\right),\quad i=3,\ldots,n-1,\\
\lbrace p_n,q_{n-2}\rbrace'&=\lbrace p_n,q_{n-3}\rbrace'=\cdots=\lbrace p_n,q_1\rbrace'=2e^{q_{n-1}-q_n}-2me^{mq_n},\\
\lbrace q_i,p_i\rbrace'&=p_i^2+2e^{q_i-q_{i+1}},\quad i=1,\ldots,n-1,\\
\lbrace q_n,p_n\rbrace'&=p_n^2+2e^{mq_n},\\
\lbrace q_{i+1},p_i\rbrace'&=e^{q_i-q_{i+1}},\\
\lbrace q_i,p_{i+1}\rbrace'&=2e^{q_{i+1}-q_{i+2}}-e^{q_i-q_{i+1}},\quad i=1,\ldots,n-2,\\
\lbrace q_{n-1},p_n\rbrace'&=2me^{mq_n}-e^{q_{n-1}-q_n},\\
\lbrace p_i,p_{i+1}\rbrace'&=-e^{q_i-q_{i+1}}(p_i+p_{i+1}).
\end{aligned}
\end{equation}
 As in the case of the 2-particle $C_2$-Toda system, see Subsection \ref{ss:C2},
 the phase spaces of the open orthogonal Toda systems are endowed with a PN strucure $(\pi,N)$, where 
\begin{equation}
N=
{\pi'}
\,
\pi
^{-1}=
\left( \begin {array}{ccc|ccc}
        \lbrace p_1,q_1\rbrace' &  \cdots &\lbrace p_n,q_1\rbrace' &\lbrace q_1,q_1\rbrace' &\cdots & \lbrace q_1,q_n\rbrace'\\
        \vdots & \cdots  & \vdots & \vdots &\cdots & \vdots\\
        \lbrace p_1,q_n\rbrace' &  \cdots &\lbrace p_n,q_n\rbrace' &\lbrace q_n,q_1\rbrace' &\cdots & \lbrace q_n,q_n\rbrace'\\
       \hline
        \lbrace p_1,p_1\rbrace'  & \cdots &\lbrace p_n,p_1\rbrace' &\lbrace p_1,q_1\rbrace' &\cdots & \lbrace p_1,q_n\rbrace'\\
        \vdots & \cdots & \vdots & \vdots &\cdots & \vdots\\
        \lbrace p_1,p_n\rbrace' & \cdots &\lbrace p_n,p_n\rbrace' &\lbrace p_n,q_1\rbrace' &\cdots & \lbrace p_n,q_n\rbrace'
        \label{eq: matrix tensor Cn}
\end{array}\right).
\end{equation}
In the following we generalize the result obtained in Section \ref{sec:Toda} to revisit the integrability of the closed 
Toda-type systems whose Hamiltonians are
\begin{equation}\label{eq: periodic Hamiltonians}
        H=\frac{1}{2}\sum_{i=1}^np_i^2+\sum_{i=1}^{n-1}e^{q_{i-1}-q_i}+e^{mq_n}+e^{-2q_1},
\end{equation}
where $m=1$ and $m=2$ correspond, respectively, to the affine Lie algebras $A_{2n}^{(2)}$ and $C_{n}^{(1)}$.
For the sake of simplicity, we initially deal with the case $m=2$ only.

We consider the same closed 2-form $\Omega_1=\d(e^{-2q_1}\d p_1)$, already used in the $C_{2}^{(1)}$ case, to 
obtain the PqN deformed structure $(\pi,\widehat{N},\widehat{\phi})$. It is easy to check that 
\begin{equation}
H_1=\frac{1}{2}\Tr(\widehat{N})=-2H
\end{equation}
and that the expression \eqref{eq: phi_expression} for $\widehat{\phi}$ still holds. 
Moreover, a straightforward computation shows that the differentials of the 1-forms
\begin{align*}
    N^{*}\left(\d q_1\right)&=\sum_{i=1}^{n}\lbrace p_i,q_1\rbrace' \d q_i+\sum_{i=1}^n\lbrace q_1,q_i\rbrace' \d p_i,\\
    N^{*}\left(\d p_1\right)&=\sum_{i=1}^{n}\lbrace p_i,p_1\rbrace' \d q_i+\sum_{i=1}^n\lbrace p_1,q_i\rbrace' \d p_i
\end{align*}
are given by the last two equations of \eqref{eq: differentials}. Therefore,
\begin{align*}
    \widehat{\phi}
    &=\Big(-4e^{-2q_1}\left(e^{q_1-q_2}-2e^{q_2-q_3}\right)\d q_1\wedge \d q_2-8e^{-2q_1}\sum_{i=3}^{n-1}\left(e^{q_{i-1}-q_i}-e^{q_i-q_{i+1}}\right)\d q_1\wedge \d q_i\\
    &\;\;\;\;-8e^{-2q_1}\left(e^{q_{n-1}-q_n}-2e^{2q_n}\right)\d q_1\wedge \d q_n+8e^{-2q_1}\sum_{i=2}^np_i\d q_1\wedge \d p_i+4p_1e^{-2q_1}\d p_1\wedge \d q_1\\
    &\;\;\;\;-2e^{-2q_1}e^{q_1-q_2}\d q_1\wedge \d q_2\Big)\wedge \d p_1
    -2e^{-2q_1}e^{q_1-q_2}\d q_1\wedge \d q_2\wedge \d p_1\\
    &=-8e^{-2q_1}\sum_{i=2}^{n-1}\left(e^{q_{i-1}-q_i}-e^{q_i-q_{i+1}}\right)\d q_1\wedge \d q_i\wedge \d p_1-8e^{-2q_1}\left(e^{q_{n-1}-q_n}-2e^{2q_n}\right)\d q_1\wedge \d q_n\wedge \d p_1\\
    &\;\;\;\;-8e^{-2q_1}\sum_{i=2}^np_i\d q_1\wedge \d p_1\wedge \d p_i\\
    &=-2\,\d H_1\wedge\Omega_1,
\end{align*}
so that condition $(a)$ of Theorem \ref{thm:involution} is satisfied. Condition $(b)$ of Theorem \ref{thm:involution} is verified mimicking the proof of item (iii) in Theorem 7 of \cite{FMOP2020}, corresponding to the classical periodic Toda lattice. We present here only the main points. From equation \eqref{ff2} it follows that 
$X_{k+1}=\widehat{N}X_k-\pi^{\sharp}\phi_{k-1}$ and so we have 
\begin{align*}
Y_k&=\widehat{N}^{k-1}X_1-X_k=\sum_{l=1}^{k-1}\left(\widehat{N}^{k-l}X_l-\widehat{N}^{k-l-1}X_{l+1}\right)\\
&=\sum_{l=1}^{k-1}\widehat{N}^{k-l-1}\left(\widehat{N}X_l-X_{l+1}\right)=\pi^{\sharp}\sum_{l=0}^{k-2}\left(\widehat{N}^*\right)^{k-l-2}\phi_l.
\end{align*}
Notice that $i_{Y_k}\Omega_1=0$ is equivalent to $\left<\d q_1,Y_k\right>=\left<\d p_1,Y_k\right>=0$ for all $k\geq 1$, and that
\begin{align}
\left<\d q_1,Y_k\right>&=\sum_{l=0}^{k-2}\left<\phi_l,\left(\widehat{N}\right)^{k-l-2}\partial_{p_1}\right>\label{eq: phi_l dp_1},
\\ 
\left<\d p_1,Y_k\right>&=-\sum_{l=0}^{k-2}\left<\phi_l,\left(\widehat{N}\right)^{k-l-2}\partial_{q_1}\right>\label{eq: phi_l dq_1}. 
\end{align}
Considering equation \eqref{eqtr}, we see that each summand of \eqref{eq: phi_l dp_1} has the form
\begin{align*}
\left<\phi_l,\widehat{N}^{k-l-2}\partial_{p_1}\right>
&=-\left<\d H_1,\widehat{N}^{k-l-2}\partial_{p_1}\right>\Tr\left(\widehat{N}^l\pi^{\sharp}\Omega_1^{\flat}\right)-2\Omega_1\left(\widehat{N}^{k-l-2}\partial_{p_1},\widehat{N}^{l}X_1\right),
\end{align*}
where 
the three terms appearing in the right-hand side of the above equation can be written as
\begin{align*}
\left<\d H_1,\widehat{N}^{k-l-2}\partial_{p_1}\right>
&=-\left<\left(\widehat{N}^*\right)^{k-l-2}\d H_1,\pi^{\sharp}\d q_1\right>=\left<\d q_1,\widehat{N}^{k-l-2}X_1\right>,\\
\Tr\left(\widehat{N}^l\pi^{\sharp}\Omega_1^{\flat}\right)
&=\left<\d q_1,\widehat{N}^{l}\pi^{\sharp}\Omega_1^{\flat}\partial _{q_1}\right>+\left<\d p_1,\widehat{N}^{l}\pi^{\sharp}\Omega_1^{\flat}\partial_{p_1}\right>=-4e^{-2q_1}\left<\d p_1,\widehat{N}^l\partial_{p_1}\right>,\\
2\Omega_1\left(\widehat{N}^{k-l-2}\partial_{p_1},\widehat{N}^lX_1\right)  &=-4e^{-2q_1}\left(\left<\d q_1,\widehat{N}^{k-l-2}\partial_{p_1}\right>\left<\d p_1,\widehat{N}^{l}X_1\right>-\left<\d q_1,\widehat{N}^lX_1\right>\left<\d p_1,\widehat{N}^{k-l-2}\partial_{p_1}\right>\right).  
\end{align*}
Thus we have
\begin{equation}
\begin{aligned}
\left<\phi_l,\widehat{N}^{k-l-2}\partial_{p_1}\right>&=4e^{-2q_1}\left(\left<\d q_1,\widehat{N}^{k-l-2}X_1\right>\left<\d p_1,\widehat{N}^{l}\partial_{p_1}\right>+\left<\d q_1,\widehat{N}^{k-l-2}\partial_{p_1}\right>\left<\d p_1,\widehat{N}^{l}X_1\right>\right.\\
&\;\;\;\;\left.-\left<\d q_1,\widehat{N}^lX_1\right>\left<\d p_1,\widehat{N}^{k-l-2}\partial_{p_1}\right>\right).\label{eq: expresssion phi_l dp_1}    
\end{aligned}
\end{equation}
Hence from equations \eqref{eq: phi_l dp_1} and \eqref{eq: expresssion phi_l dp_1} we obtain that
\begin{align*}
\left<\d q_1,Y_k\right>
&=\sum_{l=0}^{k-2}\left<\phi_l,\left(\widehat{N}\right)^{k-l-2}\partial_{p_1}\right>=4e^{-2q_1}\sum_{l=0}^{k-2}\left<\d q_1,\widehat{N}^{k-l-2}\partial_{p_1}\right>\left<\d p_1,\widehat{N}^lX_1\right>=0,
\end{align*}
since each term $\left<\d q_1,\widehat{N}^{k-l-2}\partial_{p_1}\right>$ vanishes. Indeed, it is the entry $(1,n+1)$ 
of $\widehat{N}^r$ and it is zero because the $n\times n$ upper right block of $\widehat{N}^r$ is skew-symmetric, as a consequence of $\widehat{N}^{r}\pi^\sharp=\pi^\sharp(\widehat{N}^*)^{r}$.
Analogously, one can see that \eqref{eq: phi_l dq_1} becomes
\[
\left<\d p_1,Y_k\right>=-4e^{-2q_1}\sum_{l=0}^{k-2}\left<\d q_1,\widehat{N}^{l}X_1\right>\left<\d p_1,\widehat{N}^{k-l-2}\partial_{q_1}\right>=0,
\]
since $\left<\d p_1,\widehat{N}^{k-l-2}\partial_{q_1}\right>=0$ because of the skew-symmetry of the $n\times n$ lower left block  of $\widehat{N}^r$. Therefore, the phase space of the $n$-particle periodic Toda lattice of type $C_{n}^{(1)}$ has a geometrical formulation in terms of involutive PqN manifold.
\begin{rem}
With a few modifications, considering $m=1$ in \eqref{eq: B_n & C_n brackets} 
and \eqref{eq: periodic Hamiltonians}, we also have an involutive PqN manifold formulation 
on the phase space of the 
Toda lattice for the twisted affine loop algebra $A_{2n}^{(2)}$.    
\end{rem}

\section*{Appendix B: Proof of Proposition \ref{prop:F-rel}}
\label{app:4particle-proof}

We herewith recover the relation $P_0=F_*\pi F^*$ and prove the equality $P_1=F_*\,\pi_{N_-}^{\phantom{\sharp}}\, F^*$ for generic $n$.

To start with, we notice that the block structures of the involved matrices is as follows:
\begin{equation}\label{blockF}
F_*=\left( \begin {array}{c|c} 
A & 0\\ 
\hline
0 & I
\end {array} \right)\, ,\qquad F^*=\left( \begin {array}{c|c} 
A_T & 0\\ 
\hline
0 & I
\end {array} \right),
\end{equation}
while
\begin{equation} \label{PP01} 
\pi=\left( \begin {array}{c|c} 
0& I\\ 
\hline
-I & 0
\end {array} \right)  
, \quad 
\pi_{N_-}^{\phantom{\sharp}}=
\left( \begin {array}{c|c} 
\boldsymbol{\epsilon} & D\\ 
\hline
-D & E
\end {array} \right)
,\quad
P_0=\left( \begin {array}{c|c} 
0 & A\\ 
\hline
-A_T & 0
\end {array} \right),\quad
P_1=\left( \begin {array}{c|c} 
\widetilde{A}& B\\ 
\hline
-B_T & C
\end {array} \right).
\end{equation}
Defining the periodic Kronecker $\delta$ symbol and the $\epsilon$--symbol as
\begin{equation}
\label{perdelta}
\den{k}{j}=\delta_{k\,{\small mod}\,n,\,j\,{\small mod}\,n},\,\forall k,j\geq 0
,
\qquad \epsilon(\ell)=\left\{
\begin{array}{l}
1 \quad\ \ \text{if } \ell>0\\
0 \quad\ \ \text{if } \ell=0\\
-1\quad \text{if } \ell<0
\end{array}\right.\, ,
\end{equation}
the matrix elements of the submatrices entering Eqs.\ (\ref{blockF}) and (\ref{PP01}) are expressed as:
\begin{equation}
\label{PP01matel}
\begin{array}{ll}
A_{k,j} =a_k\delta_{k,j}-a_k\den{k}{j-1} , &\boldsymbol{\epsilon}_{k,j}=\epsilon(j-k),\qquad \qquad  {D}_{k,j}=p_k\delta_{k,j}\equiv b_k \delta_{k,j},\\
E_{k,j}=a_k\den{k}{j-1}-a_j\den{k-1}{j},  & B_{k,j}=a_k\, b_k\delta_{k,j}-a_k\, b_j \delta^{(n)}_{k,j-1},\\
 \widetilde{A}_{k,j}=-a_ka_j\den{k}{j-1}+a_ka_j\den{k}{j+1},  & C_{k,j}=a_k\den{k}{j-1}-a_j\den{k}{j+1}.
\end{array}
\end{equation}
At first we remark that 
\begin{equation}\label{FP0ok}
F_*\pi F^*=\left( \begin {array}{c|c} 
0 & A\\ 
\hline
-A_T & 0
\end {array} \right), 
\end{equation}
which is just the relation between $\pi$ and $P_0$.

For the second structure, we observe that 
\begin{equation}\label{checkP1} 
F_*\,\pi_{N_-}^{\phantom{\sharp}}\, F^*=\left( \begin {array}{c|c} 
A\boldsymbol{\epsilon} A_T & AD\\ 
\hline
-DA_T & E
\end {array} \right),
\end{equation}
so by looking at (\ref{PP01matel}) the only non-trivial equality to be proven is
\begin{equation}
\label{nteqapp}
A\boldsymbol{\epsilon}A_T=\widetilde{A}.
\end{equation}
This holds true since, as straightforward computations show, one has that
\begin{equation}
\label{conto1}
\left(A\boldsymbol{\epsilon}A_T\right)_{\ell,k}=a_\ell\, a_k\, \left(2\epsilon(k-\ell)-\epsilon(k-\ell-1)-\epsilon(k-\ell+1)\right),
\end{equation}
and
\begin{equation}
\label{conto2}
\left(2\epsilon(k-\ell)-\epsilon(k-\ell-1)-\epsilon(k-\ell+1)\right)=\den{k}{\ell+1}-\den{k}{\ell-1}.
\end{equation}


\par\medskip\noindent
{\bf Acknowledgments.} 
We thank Murilo do Nascimento Luiz, Franco Magri, Giovanni Ortenzi, and Giorgio Tondo for useful discussions. MP thanks the ICMC-USP, {\em Instituto de Ci\^encias Matem\'aticas e de Computa\c c\~ao\/} of the University of S\~ao Paulo, and the Department of Mathematics and its Applications of the University of Milano-Bicocca for their hospitality, the Funda\c c\~ao de Amparo \`a Pesquisa do Estado de S\~ao Paulo - FAPESP, Brazil, for supporting his visit in 2023 to the ICMC-USP with the grant 2022/02454-8, and the University of Bergamo, for supporting his visit in 2024 to the ICMC-USP within the program {\em Outgoing Visiting Professors}.
This project has received funding from the European Union's Horizon 2020 research and innovation programme under the 
Marie Sk{\l}odowska-Curie grant no 778010 {\em IPaDEGAN} as well as by the Italian PRIN 2022 (2022TEB52W) - PE1 - project {\em The charm of integrability: from nonlinear waves to random matrices}. All authors gratefully acknowledge the auspices of the GNFM Section of INdAM under which part of this work was carried out. 
We are grateful to the anonymous referee, whose suggestions helped us to substantially improve the content 
of our manuscript.


\medskip\noindent{\bf Data availability.} Data sharing was not applicable to this article as no datasets were generated or analyzed during the current study.


\medskip\noindent{\bf Conflict of interest.} On behalf of all authors, the corresponding author states that there is no conflict of interest.




\end{document}